\documentclass[a4paper,11pt]{article}

\usepackage{authblk}
\usepackage{graphicx}

\usepackage{color,enumerate}
\usepackage{amsfonts}
\usepackage{amssymb}
\usepackage{amsthm}
\usepackage{amsmath}
\usepackage{graphicx}
\usepackage{graphics}
\usepackage{a4wide}
\usepackage{url}

\def\Exp{{\mathbb{E}}}

\newcommand{\R}{{\mathbb R}}

\theoremstyle{plain}

\theoremstyle{remark}

\theoremstyle{plain}
\newtheorem{Thm}{Theorem}

\newtheorem{Post}{Postulate}

\newcommand{\ceq}{\stackrel{+}{=}}

\title{Algorithmic independence of initial condition and dynamical law 
in
thermodynamics and causal inference}

\author[1]{Dominik Janzing}
\author[2,3]{Rafael Chaves}
\author[1]{Bernhard Sch\"olkopf}
\affil[1]{\footnotesize{Max Planck Institute for Intelligent Systems,
Spemannstr. 38,
72076 T\"ubingen, Germany}}
\affil[2]{\footnotesize{Institute for Physics \& FDM, University of Freiburg, 79104 Freiburg, Germany}}
\affil[3]{\footnotesize{Institute for Theoretical Physics, University of Cologne, 50937 Cologne, Germany}}

\date{December 7, 2015}

\begin{document}

\maketitle

\begin{abstract}
We postulate a principle stating that the initial condition of a physical system is typically algorithmically independent of the dynamical law. We argue that this links thermodynamics and causal inference. On the one hand, it entails behaviour that is similar to the usual arrow of time. On the other hand, it motivates a statistical asymmetry between cause and effect that has recently postulated in the field of causal inference, namely, that the probability distribution $P_{\rm cause}$ contains no information about the conditional distribution $P_{{\rm effect}|{\rm cause}}$ and vice versa, while $P_{\rm effect}$ may contain information about $P_{{\rm cause}|{\rm effect}}$.
\end{abstract}



Drawing causal conclusions from statistical data is at the heart of modern scientific research. While it is generally accepted that {\it active} interventions to a system (e.g. randomized trials in medicine) reveal causal relations, statisticians have widely shied away from drawing causal conclusions from {\it passive} observations. Meanwhile, however, the increasing interdisciplinary field of causal inference has shown that also the latter is possible --even without information about time order-- if appropriate assumptions that link causality and statistics are made~\cite{Spirtes1993,Pearl:00,pearl2014probabilistic}, with applications in biology~\cite{Friedman2004}, psychology~\cite{Jackson}, and economy~\cite{Moneta}. More recently, also foundational questions of quantum physics have been revisited in light of the formal language and paradigms of causal inference~\cite{Fritz2012, Leifer2013,Spekkens2015,Chaves2015a,Chaves2015b,Henson14,Ried15}.

Remarkably, recent results from causal inference have also provided new insights about the thorny issue of the arrow of time. Contrary to a wide-spread belief, the joint distribution $P_{X,Y}$ of two variables $X$ and $Y$ sometimes indicates whether $X$ causes $Y$ or vice versa \cite{Mooij2014}. More conventional methods rely on conditional independencies and thus require statistical information of at least $3$ observed variables \cite{Spirtes1993,Pearl:00}. The idea behind the new approach is that if $X$ causes $Y$, $P_X$ contains no information about $P_{Y \vert X}$ and vice versa.
Like the asymmetries between cause and effect, similar asymmetries between past and future are also manifest even in stationary time series \cite{ICMLTimeSeries} which can sometimes be used to infer the direction of empirical time series (e.g. in finance or brain research) or to infer the time direction of movies~\cite{Pickup}. Altogether, these results suggest a deeper connection for the asymmetries between cause vs.~effect and past vs.~future. In particular, a physical toy model relating such asymmetries to the usual thermodynamic arrow of time has been proposed~\cite{TimeSeriesThermo}.

Motivated by these insights, we propose a foundational principle for both types of asymmetries, basically stating that the initial state of a physical system and the dynamical law to which it is subjected to should be algorithmically independent. As we show, it implies for a closed system the non-decrease of physical entropy if the latter is identified with the Kolgomorov complexity, in agreement with Ref.~\cite{ZurekKol}.
Similarly, the independence of $P_{\mathrm{cause}}$ and $P_{\mathrm{effect}\vert\mathrm{cause}}$ postulated in causal inference naturally follows, if we identify cause and effect with the initial and final states of a physical system. Moreover, we employ our principle to understand open system dynamics and apply it to a toy model representing typical cause-effect relations.

\section{Algorithmic randomness in thermodymamics} 
The algorithmic randomness (or Kolmogorov complexity) $K(s)$ of a binary string $s$ is defined as the length of its shortest compression. More precisely, $K(s)$ is the length of the shortest program on a universal Turing machine that generates $s$ and stops then \cite{KolmoOr,ChaitinF}. In a seminal paper, Bennett \cite{BennettThermoReview}  proposed to consider $K(s)$ as the {\it thermodynamic entropy} of a microscopic state of a physical system when $s$ describes the latter with respect to some standard binary encoding after sufficiently fine discretization of the phase space. This assumes an `internal' perspective (followed here), where the microscopic state is perfectly known to the observer. Although $K(s)$ is in principle uncomputable, it can be estimated from the Boltzmann entropy in many-particle systems, given that the microscopic state is typical in a set of states satisfying some macroscopic constraints \cite{BennettThermoReview,ZurekKol}. That is, in practice one needs to rely on more conventional definitions of physical entropy.

From a theoretical and fundamental perspective, however, it is appealing to have a definition of entropy that neither relies on missing knowledge like the statistical Shannon/von-Neumann entropy \cite{Cover,NC}
nor on the separation between microscopic vs.~macroscopic states --which becomes problematic on the mesoscopic scale-- like the Boltzmann entropy \cite{Jaynes65}. For imperfect knowledge of the microscopic state, Zurek  \cite{ZurekKol} considers thermodynamic entropy as the sum of statistical entropy and Kolmogorov complexity \cite{CavesKolmo}, which thus unifies the statistical and the algorithmic perspective of physical entropy.

To discuss how $K(s)$ behaves under Hamiltonian dynamics, notice that dynamics on a continuous space is usually not compatible with discretization, which immediately introduces also statistical entropy in addition to the algorithmic term --particularly for chaotic systems \cite{CavesChaos}-- in agreement with standard entropy increase by coarse-graining \cite{Balian,Zeh01}. Remarkably, however, $K(s)$ can also increase by applying a one-to-one map $D$ on a {\it discrete} space \cite{ZurekKol}. Then $K(s)+K(D)$ is the tightest upper bound for $K(D(s))$ that holds for the general case. For a system starting in a simple initial state $s$ and evolving by the repeated application of some simple map $\tilde{D}$, the description of $s_t:=D(s):=\tilde{D}^t(s)$ essentially amounts to describing $t$ and Zurek derives logarithmic entropy increase until the scale of the recurrence time is reached \cite{ZurekKol}. Although logarithmic growth is rather weak \cite{CavesChaos}, it is worth mentioning that the arrow of time here emerges from assuming that the system starts in a {\it simple} state. We will later argue
that this is just a special case of the general assumption that it starts in a state that is independent of $D$. The fact that $K$ depends on the Turing machine could arguably spoil its use in physics, but, in the spirit of Deutsch's idea that the laws of physics determine the laws of computation \cite{Deutsch}, future research may define a 'more physical version' of $K$ by computation models whose elementary steps directly use physically realistic particle interactions, see e.g. \cite{KontiKomplex,Martineffizient,ErgodicQutrits}. Moreover, {\it quantum} thermodynamics \cite{QTbook} should rather rely on {\it quantum} Kolmogorov complexity \cite{MoraKraus}.

\section{Algorithmic randomness in causal inference} 
Reichenbach's principle \cite{Reichenbach1956} states that every statistical dependence between two variables $X,Y$ must involve some sort of causation: either direct causation ($X$ causes $Y$ or vice-versa) or a common cause for both $X$ and $Y$. Conversely, variables without causal relation are statistically independent. Surprisingly, however, statistical ensembles are not necessarily required for drawing causal conclusions. As argued in Ref.~\cite{Algorithmic}, two binary words $x,y$, representing two causally disconnected objects should also be algorithmically independent, where algorithmic dependence is measured by \cite{ChaitinF}
\begin{eqnarray*}
I(x:y)&:=& K(x)+K(y)-K(x,y)\ceq K(x)-K(x|y^*)\ceq K(y)-K(y|x^*)\,.
\end{eqnarray*}
As common in algorithmic information theory \cite{Vitanyi08}, the plus sign above (in)equalities or signs denote errors of order $O(1)$, i.e., they do not depend on the length of the binary strings $x,y$. Further, $K(x|y^*)$ denotes the conditional Kolmogorov complexity, i.e., the length of the shortest program that generates $x$ from $y^*$ (denoting the shortest compression of $y$). Thus, $I(x:y)$ is the number of bits saved when $x$ and $y$ are compressed jointly rather than independently, or, equivalently, the number of bits that the description of $x$ can be shortened when the shortest description of $y$ is known. 
Remarkably, algorithmic information can be used
to infer whether $X$ causes $Y$ (denoted by $X\rightarrow Y$) or $Y$ causes $X$ from their joint distribution, given that exactly one of the alternatives is true \footnote{The difficult question how well-defined causal directions emerge in physical
  systems where {\protect \it interactions} actually imply mutual influence is
  discussed for a toy model in Ref.~\cite {MitArmen}.}. If $P_{\rm cause}$ and $P_{{\rm effect}|{\rm cause}}$ are `independently chosen by nature' and thus causally unrelated, their algorithmic mutual information is also zero \cite{Algorithmic}. Here, it is implicitly assumed, for sake of simplicity, that the joint probability distribution is known and that it is computable.
This is, for instance the case if we replace `distribution' with the empirical frequencies
obtained via sufficiently fine binning of some observed data set with sufficiently large sampling.

First, consider linear non-Gaussian additive noise models \cite{Kano2003}. Then, the joint distribution $P_{X,Y}$ of two random variables $X,Y$ admits the linear model
\begin{equation}\label{eq:lingam}
Y=\alpha_X X + N_Y\,,
\end{equation}
where $\alpha_X\in \R$ and $N_Y$ is an unobserved noise term that is statistically independent of $X$ \footnote{Note that two variables $Z,W$ are called statistically independent if
  $p(Z,W)=p(Z)p(W)$, which is stronger than being uncorrelated, i.e.,
  $\Exp [ZW]=\Exp [Z] \Exp [W]$.}. Whenever $X$ or $N_Y$ is non-Gaussian, it follows that for every model of the form $X=\alpha_Y Y+N_X$, the noise term $N_X$ and $Y$ are statistically dependent, although they may be uncorrelated. That is, except for Gaussian variables, a linear model with independent noise can hold in one direction at most. Within that context, Ref.~\cite{Kano2003} infers the direction with additive independent noise to be the causal one. To justify this reasoning, Ref.~\cite{Steudel_additive_noise} argues that whenever \eqref{eq:lingam} holds, the densities of $P_Y$ and $P_{X|Y}$ are related by the differential equation
\[
 \frac{\partial^2}{\partial y^2} \log p(y)= -\frac{\partial^2}{ \partial y^2} \log p(x|y)  -\frac{1}{\alpha_X}\frac{\partial^2}{ \partial x \partial y}  p(x|y)\,.
\]
Therefore, knowing $P_{X|Y}$ enables a short description of $P_Y$.
Whenever $P_Y$ has actually high description length, we       
 should thus reject $Y\rightarrow X$ as a causal explanation.

Second, consider the information-geometric approach to causal inference \cite{deterministic,Info-Geometry}. Assume that $X$ and $Y$ are random variables with values in $[0,1]$, deterministically related by $Y=f(X)$ and $X=f^{-1}(Y)$, where $f$ is a monotonically increasing one-to-one mapping of $[0,1]$. Ref.~\cite{deterministic} proposes to infer $X\rightarrow Y$ whenever the derivative of $f$ and the probability density of $P_X$ satisfy
\begin{equation}\label{eq:igci}
\int_0^1 \log f'(x) p(x) dx < 0\,,
\end{equation}
and infer $Y\rightarrow X$ if the integral above is positive.
 The reason is that negative values are typical if $P_X$ is generated by a random procedure that is independent of $f$ \cite{Info-Geometry}, while positive values require $p(x)$ to be particularly large in regions where $f$ has large slope. Then, knowing $f$ (here describing the conditional distribution $P_{Y|X}$) constrains the distribution of $X$ \cite{JustiIGCI} and thus makes its description shorter.

\section{A common root for thermodynamics and causal inference} 
To provide a unifying foundation connecting thermodynamics and causal inference we propose the following postulate:
\begin{Post}[algorithmic independence between input and mechanism]\label{post:aim}
Let $s$ be the initial state of a system and $M$ a mechanism
mapping $s$ to a final state.
If the preparation of $s$ is done without any
knowledge of $M$, then
\begin{equation}\label{eq:icm}
I(s:M)\ceq 0\,.
\end{equation}
That is, knowing $s$ does not enable a shorter description of $M$
and vice versa.
\end{Post}
The postulate is entailed by the assumption that there is no algorithmic dependence in nature without an underlying causal relation. By overloading notation, we have identified mechanism and state with their encodings into binary strings. 
Generalizations of algorithmic mutual information for infinite strings can be found in \cite{Levin}, which then allows to apply Postulate~\ref{post:aim} to continuous physical state spaces. Here, however, we consider finite strings describing states after sufficiently fine discretizations of the state space instead,
neglecting issues from chaotic systems \cite{CavesChaos} for sake of conciseness.

\section{Dynamics of closed physical systems} 
To show the first consequence of Postulate \eqref{post:aim}, consider a physical system whose state space is a finite set $S$. Assuming that the dynamics $D$ is a bijective map of $S$, it follows that the entropy cannot decrease:
\begin{Thm}[no entropy decrease]\label{thm:entropyincr}
If the dynamics of a system is a one-to-one mapping  $D$ of a discrete
set $S$ of states then Postulate~\ref{post:aim} implies that
the Kolmogorov complexity can never decrease when applying $D$
to the initial state $s$, i.e.,
\begin{equation}\label{eq:entropyincr}
K(D(s))\stackrel{+}{\geq} K(s)\,.
\end{equation}
\end{Thm}
\begin{proof}
Algorithmic independence of $s$ and $D$ amounts to $K(s)\ceq K(s|D^*)$. For any known bijection $D$, $s$ can be computed from $D(s)$ and vice versa implying that $K(s|D^*)\ceq K(D(s)|D^*)$. Thus, $K(s)\ceq K(s|D^*) \ceq K(D(s)|D^*) \stackrel{+}{\leq} K(D(s))$, concluding the proof.
\end{proof}
$D$ may also be considered as the $t$-fold concatenation of the same map $\tilde{D}$, where $\tilde{D}$ has negligible description length. Then, $I(s:D)\ceq 0$ amounts to $I(s:t)\ceq 0$. Theorem~\ref{thm:entropyincr} then implies that $K(\tilde{D}^t(s)) \stackrel{+}{\geq} K(s)$ whenever $t$ and $s$ are algorithmically independent. That is, while Ref.~\cite{ZurekKol} derives entropy increase for a simple initial state $s$, we have derived it for all states $s$ that are independent of $t$.

To further illustrate Theorem~\ref{thm:entropyincr}, consider a toy model of a physical system consisting of $n \times m$ cells, each being occupied or not with a particle, see Figure~\ref{fig:caordered}. Its state is described by a binary word $s$ with $nm$ digits. For generic $s$, we have $K(s)\approx nm$, while Figure~\ref{fig:caordered}, left, shows a simple state where all particles are in the left uppermost corner containing $k\times l$ cells. A description of this state $s$ consists essentially of describing $k$ and $l$ (up to negligible extra information specifying that $k$ and $l$ describe the size of the occupied region), which requires $\log_2 k+\log_2 l$ bits.
\begin{figure}
\centerline{
\includegraphics[width=0.6\columnwidth]{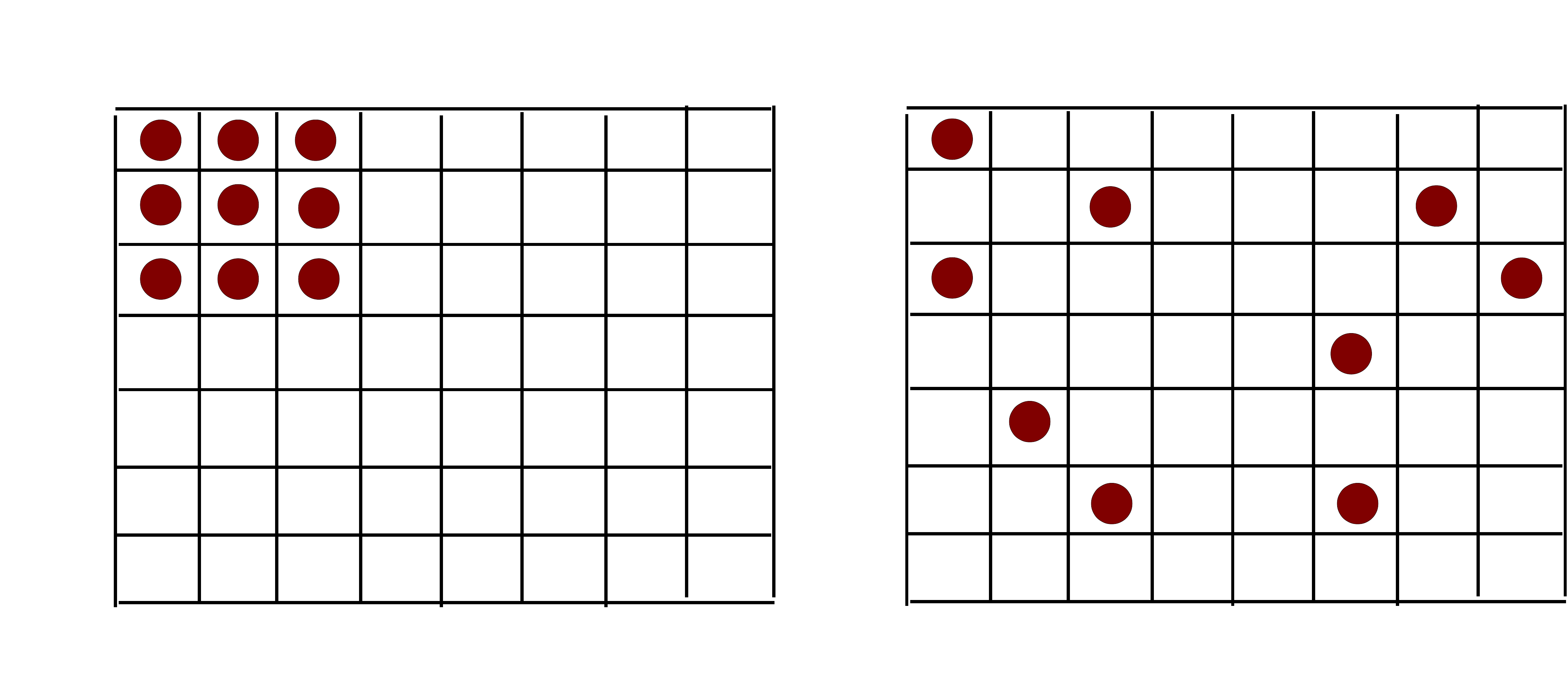}
}
\caption{\label{fig:caordered} Left: cellular automaton starting in a state with small description length. Right: the state $s'$ obtained from $s$ by application of the dynamical law $D$.
}
\end{figure}
Consider now that the dynamical evolution $D$ transforms $s$ into $s'=D(s)$ where $s'$ looks `more generic', as shown in Figure~\ref{fig:caordered}, right.
In principle, we cannot exclude that $s'$ is equally simple as $s$ due to some non-obvious pattern.  However, excluding this possibility as unlike, Theorem~\ref{thm:entropyincr} rules out any scenario where $s'$ is the initial state and $s$ the final state of {\it any} bijective mapping $D$ that is algorithmically independent of $s'$. The transition from $s$ to $s'$ can be seen as a natural model of a mixing process of a gas, as described by popular toy models like lattice gases \cite{Frisch87}. These observations are consistent with standard results of statistical mechanics saying that mixing is the typical behaviour, while de-mixing requires some rather specific tuning of microscopic states. Here we propose to formalize 'specific' by means of algorithmic dependencies between the initial state and the dynamics. This view does not necessarily generate novel insights for typical scenarios of statistical physics, but it introduces a link to crucial concepts in the field of causal inference.

\section{Dynamics of open systems}
Since applying \eqref{eq:icm} to closed systems reproduces the standard thermodynamic law of non-decrease of entropy, it is appealing to state algorithmic independence for closed system dynamics only and then obtain conditions under which the independence for open system follows.

Let $D$ be a one-to-one map transforming the initial joint state $(s,e)$ of system and environment into the final state $(s',e')$. For fixed $e$, define the open system dynamics $M: s\mapsto s'$. If $s$ is algorithmically independent of the pair $(D,e)$ (which is true, for instance when $K(e)$ is negligible and $s$ and $D$ are independent),
independence of $s$ and $M$ follows because algorithmic independence of two strings $a,b$ implies independence of $a,c$ whenever $c$ can be computed from $b$ via a program of length $O(1)$ \cite{Algorithmic}.

Further, we can extend the argument above to statistical ensembles: consider $n$ systems with identical state space $S$, each coupled to an environment with identical state space $E$. Let $(s_j,e_j)\in S\times E$ be the initial state of the $j$th copy. Then $s^n:=(s_1,\dots,s_n)$ and $e^n:=(e_1,\dots,e_n)$ define empirical distributions $P_S$ on $S$ and $P_E$ on $E$, respectively, counting the number of occurrences in the respective $n$-tuple. For fixed $P_E$, the dynamics $D$ acting on each copy $S\times E$ defines a conditional distribution for the final state $s'\in S$ of one copy, given its initial state $s\in S$ via
\[
P_{S'|S}(s'|s):=\sum_e P_E(e)\,,
\]
where the sum runs over all $e$ with $D(s,e)=(s',e')$ for some $e'\in E$. If $s^n$ is algorithmically independent of $(D,e^n)$, we can conclude that $P_S$ and $P_{S'|S}$ are algorithmically independent, because they are derived from two algorithmically independent objects via a program of length $O(1)$.
Defining the variable '${\rm cause}$' by the initial state of one copy $S$ and '${\rm effect}$' as the final state, we have thus derived the algorithmic independence of $P_{\rm cause}$ and $P_{{\rm effect}|{\rm cause}}$. Notice that it is not essential in the reasoning above that cause and effect describe initial and final states of the same physical system, one could as well consider a tripartite instead of a bipartite system.

\section{Physical toy model for a deterministic non-linear cause-effect relation} 
To illustrate how the independence of  $P_{\rm cause}$ and $P_{{\rm effect}|{\rm cause}}$ is inherited from algorithmic independence of initial state and dynamics of a closed system,
Figure~\ref{fig:wall}
\begin{figure}
\centerline{
\includegraphics[width=0.6\columnwidth]{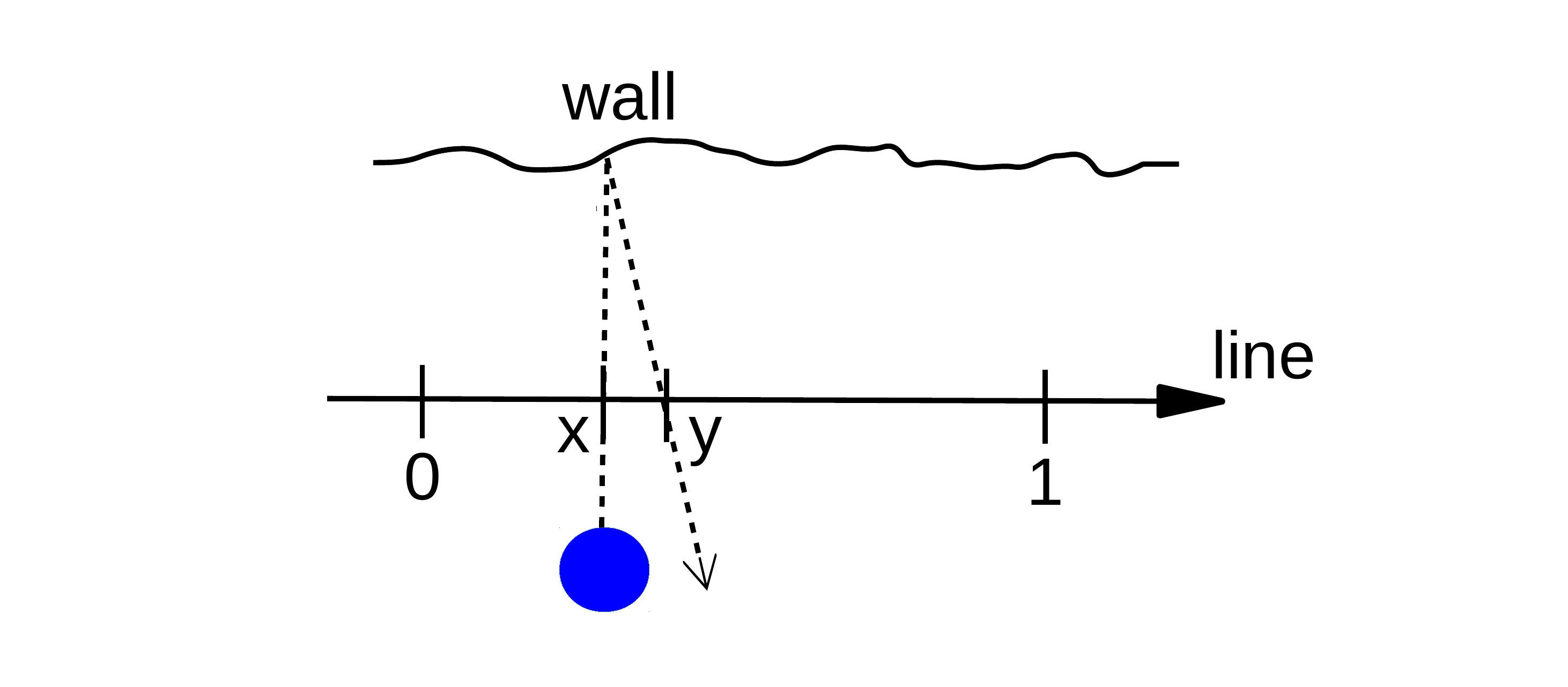}
}
\caption{\label{fig:wall} Physical system generating a non-linear deterministic causal relation: A particle travelling towards a structured wall
with momentum orthogonal to the wall, where it is backscattered in a slightly different direction. $x$ and $y$ denote the positions where $P$ crosses a vertical line before and after the scattering process, respectively.}
\end{figure}
shows a simple 2-dimensional system with a particle $P$ travelling towards a wall $W$ perpendicular to the momentum of $P$.
$P$ crosses a line $L$ parallel to $W$ at some position $x \in [0,1]$. Let the surface of $W$ be structured such that $P$ hits the wall with an incident angle that depends on its vertical position. Then $P$ crosses $L$ again at some position $y$. Assume that $L$ is so close to $W$ that the mapping $x\mapsto y=:f(x)$ is one-to-one. Also, assume that $0$ is mapped to $0$ and $1$ to $1$ and the experiment is repeated with particles having the same momenta but with different positions
such that $x$ is distributed according to some probability density $p(x)$. It is plausible that the initial distribution of momenta and positions do not contain information about the structure of $W$, while the final one does. Due to Theorem~\ref{thm:entropyincr}, the scattering process increases the Kolmogorov complexity of the state. Further, this process is thermodynamically irreversible for every thermodynamic machine that has no access to the structure of $W$.

We now focus on a restricted aspect of this physical process, namely the one leading from $p(x)$ to $p(y)$ via the function $f$, so we can directly apply the information-geometric approach to causal inference \cite{deterministic}. Again, the initial state $p(x)$ is algorithmically independent of $f$, while $p(y)$ contains information about $f$ (and thus about $f^{-1}$). Hence, algorithmic dependence indicates the time (and causal) direction of the process. Notice that the Shannon entropy of $p(y)$ is smaller than the entropy of $p(x)$ because the function $f$ will typically make the distribution less uniform \cite{deterministic}, while making it more uniform requires fine-tuning $f$ relative to $p(x)$ \footnote{However, due to Liouville's theorem, the Shannon entropy of the probability distribution in phase space (position and momenta) is certainly conserved by the scattering process.}. This could lead to misleading conclusions such as inferring the time direction from $p(y)$ to $p(x)$. This is of particular relevance in scenarios where simple criteria like entropy increase/decrease are inapplicable, at least without accounting for the description of the entire physical system (that often may be not available, e.g., if the momentum of the particle is not measured). The example above suggests how the algorithmic independence provides a new tool for the inference of time direction in such scenarios.

\section{Discussion} 
Already Reichenbach 
linked asymmetries between cause and effect to the arrow of time when he argued that the statistical dependence patterns
induced by causal structures $X\leftarrow  Z \rightarrow Y$ (common cause) vs.~$X\rightarrow Z \leftarrow Y$ (common effect) naturally emerges from the time direction of appropriate mixing processes \cite{Reichenbach1956}.
In this work we provide a new foundational principle describing additional asymmetries that appear when algorithmic rather than only statistical information is taken into account. As a consequence it follows naturally the non-decrease of algorithmic entropy and the independence between $P_{\mathrm{cause}}$ and $P_{\mathrm{effect}\vert\mathrm{cause}}$, thus providing further relations between thermodynamics and causal inference.

Intuitively, our principle resembles the standard way to obtain Markovian dynamics of open systems, coupling a system to a statistically independent environment \cite{Lindblad1976}. In this sense, our principle can be understood as a notion of Markovianity that is stronger in two respects: first, the initial state of the system is not only statistically but also algorithmically independent of the environment, and second, it is also algorithmically independent of the dynamical law. It thus provides a useful new rationale for finding the most plausible causal explanation for given observations arising in study of open systems. Furthermore, given the recent connections between the phenomenon of quantum nonlocality \cite{Brunner} with algorithmic information \cite{Poh,Bendersky} and causality~\cite{Fritz2012, Leifer2013,Spekkens2015,Chaves2015a,Chaves2015b,Henson14,Ried15}, our results may also point new directions for research in the foundations of quantum physics.


\paragraph{Acknowledgments}
The authors would like to thank Johan \r{A}berg and Philipp Geiger for helpful remarks on an earlier version of the manuscript. RC acknowledges financial support from the Excellence Initiative of the German Federal and State Governments (Grants ZUK 43 \& 81), the US Army Research Office under contracts W911NF-14-1-0098 and W911NF-14-1-0133 (Quantum Characterization, Verification, and Validation), the DFG (GRO 4334 \& SPP 1798).

\end{document}